\documentclass[sigconf]{acmart}

\AtBeginDocument{%
  }

\copyrightyear{2018}
\acmYear{2018}
\acmDOI{XXXXXXX.XXXXXXX}
\acmConference[ACM SIGSPATIAL '25]{Make sure to enter the correct
  conference title from your rights confirmation email}{June 03--05,
  2018}{Woodstock, NY}
\acmISBN{978-1-4503-XXXX-X/2018/06}

\usepackage{algorithm, algpseudocode}
\usepackage{varwidth}
\usepackage{booktabs}
\usepackage{array} 
\usepackage{amsthm}
\usepackage{tikz}
\usepackage{hyperref}
\newtheorem{definition}{Definition}
\newtheorem{theorem}{Theorem}
\usepackage{graphicx}
\usepackage{subcaption}  
\usepackage{enumitem} 
\graphicspath{{images/}}

\begin{document}

\title{PolyMinHash: Efficient Area-Based MinHashing of Polygons for Approximate Nearest Neighbor Search}

\author{Alima Subedi}
\email{asbmr@mst.edu}
\affiliation{%
  \institution{Missouri S\&T}
  \city{Rolla}
  \state{Missouri}
  \country{USA}
}

\author{Sankalpa Pokharel}
\email{sp5rt@mst.edu}
\affiliation{%
  \institution{Missouri S\&T}
  \city{Rolla}
  \country{USA}}

  \author{Satish Puri\footnotemark[1]}
\email{satish.puri@mst.edu}
\affiliation{%
  \institution{Missouri S\&T}
  \city{Rolla}
  \country{USA}}

\renewcommand{\shortauthors}{Subedi et al.}

\begin{abstract}
Similarity searches are a critical task in data mining. As data sets grow larger, exact nearest neighbor searches quickly become unfeasible, leading to the adoption of approximate nearest neighbor (ANN) searches. ANN has been studied for text data, images, and trajectories. However, there has been little effort to develop ANN systems for polygons in spatial database systems and geographic information systems. We present PolyMinHash, a system for approximate polygon similarity search that adapts MinHashing into a novel 2D polygon-hashing scheme to generate short, similarity-preserving signatures of input polygons. Minhash is generated by counting the number of randomly sampled points needed before the sampled point lands within the polygon's interior area, yielding hash values that preserve area-based Jaccard similarity. 


We present the tradeoff between search accuracy and runtime of our PolyMinHash system. Our hashing mechanism reduces the number of candidates to be processed in the query refinement phase by up to 98\% compared to the number of candidates processed by the Brute Force algorithm.

\end{abstract}



\maketitle

\footnotetext[1]{This work is partially supported by NSF grants \#2344585, \#2344578 and Taylor Geospatial Institute managed by Saint Louis University.}

\section{Introduction}\label{intr}
Similarity search plays a crucial role in data mining.
The ability to tell if two objects are similar and identify whether an object is closest to a query object is vital in some fundamental algorithms in data mining, such as clustering, information retrieval, and recommender systems.
Similarity searches are applied in multiple fields, such as pathology, solar flares~\cite{pillai2013filter}, and Geographic Information Systems (GIS), including geospatial intelligence, for analyzing the geometries (shapes) of
polygons and trajectories. In GIS, the geometries are represented as a list of coordinates (latitude, longitude). An example of shape similarity from the GIS domain is finding a lake that is similar in shape and area to Lake Michigan. In solar physics, to predict solar flares, the query object and the data set comprise polygons that represent solar events~\cite{pillai2013filter}. In digital pathology~\cite{wang2012accelerating}, tissues are represented as polygons for tumor diagnosis, and the Jaccard distance (ratio of intersection to union areas) is used to make similarity comparisons. 

As the size of datasets grows, the ability to perform these searches quickly and efficiently becomes increasingly critical to the usability of the system in practical applications. The Brute-Force approach requires scanning the entire data set to compare it to the query object, which is a linear-time operation per query. For polygonal data, comparing the query polygon with one of the input polygons in the database using the Jaccard similarity $\mathcal{J}_s$ requires expensive computational geometry operations to compute the geometric intersection and union. To find exact nearest neighbors for a query geometry, the geometric Jaccard distance is used, which is derived from geometric similarity (1-$\mathcal{J}_s$). Using this distance metric, any two similar polygons have a small distance between them. An example of ANN search is given in Fig.~\ref{fig:exq}.



Our proposed PolyMinHash system performs a search for the nearest neighbors.

\begin{definition}[Nearest Neighbor Search]
	Given a set $P$ of polygons and the Jaccard distance metric $d$, return a set $S \subset P$ such that for some query $q$, $d(q, s) \leq d(q, p)$ for all $s \in S$ and $p \in P - S$.
\end{definition}


Locality sensitive hashing (LSH) and MinHashing are theoretically robust methods for ANN search~\cite{leskovecMiningMassiveDatasetsa}. There are many LSH and MinHashing algorithms in the literature for a variety of datasets, such as feature vectors~\cite{shrivastavaSimpleEfficientWeighted2016}, text ~\cite{leskovecMiningMassiveDatasetsa}, and images. Polygons represented as variable-length vectors do not lend themselves directly to MinHashing algorithms. To the best of our knowledge, for polygonal shape data, MinHashing algorithms do not exist that directly convert a variable length geometry to area-preserving MinHash signatures that are generally much shorter than original polygons. 

Our proposed MinHashing algorithm is based on a variation of rejection sampling proposed earlier for feature vectors~\cite{shrivastavaSimpleEfficientWeighted2016} and Monte Carlo method for statistical sampling. In order to compute the MinHash signature of a polygon, we generate 2D points as samples and test for inclusion/exclusion w.r.t. its area. We have shown that the number of attempts needed by a random 2D point to land inside a polygon in the rejection sampling process preserves the Jaccard similarity when a pair of polygons are compared for approximate similarity. In this work, we do not consider scale and rotational invariance for polygons.

\begin{figure}
	\centering 
       \includegraphics[width=1.0\linewidth]{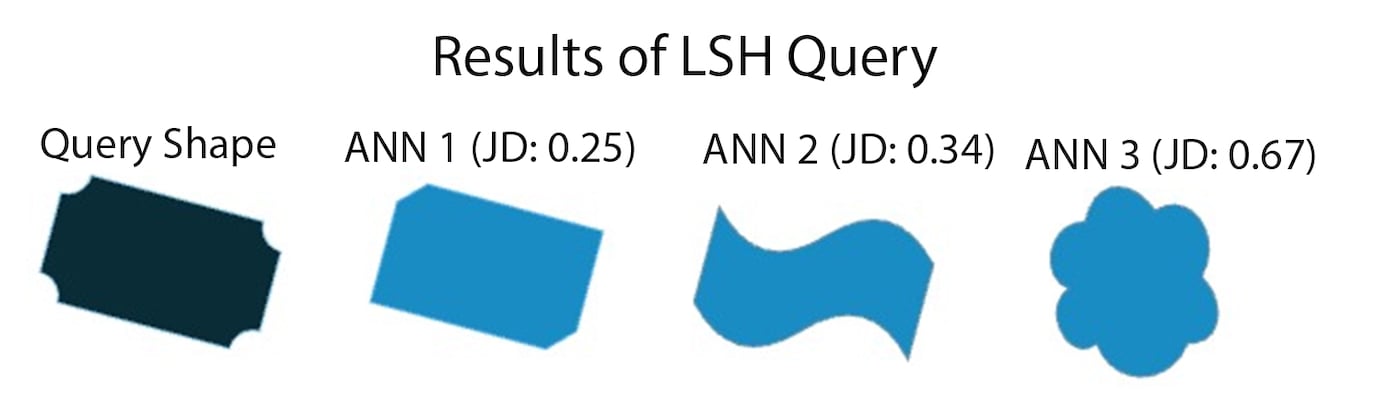}
       \caption{An example of a K-Approximate nearest-neighbor (K-ANN) query and its results for K=3. JD is the Jaccard distance.}
	\label{fig:exq}
\end{figure}

\section{Background and Related Work}\label{prwo}

Many data structures, such as the K-D tree, Voronoi diagram, similarity search tree (SS tree), R-tree, etc., have been designed for distance-based similarity search. Some data structures like the K-D tree and VP-tree (metric-based indexing) support both exact and approximate searches. These trees work effectively for low-dimensional vectors; however, they suffer from the curse of dimensionality when applied to higher-dimensional vectors~\cite{pan2024survey}.

Our proposed work is distinct from other related works along three dimensions - 1) type of datasets used for evaluation, 2) distance metric, and 3) theoretical work vs practical system. 

{\bf Datatype and Distance Metric}: Related work on ANN research is based on a variety of datasets such as trajectories~\cite{astefanoaeiMultiresolutionSketchesLocality2018}, curves~\cite{DBLP:journals/corr/DriemelS17}, images~\cite{grauman2004fast} and text\cite{leskovecMiningMassiveDatasetsa}. The LSH method presented for the trajectories in~\cite{astefanoaeiMultiresolutionSketchesLocality2018} is not suitable for polygons because the trajectories are open, but the polygons are closed. Therefore, the concept of an area of a polygon is not captured by the LSH methods used for trajectories. 

Polygon similarity has been studied in computer graphics and computer vision. First of all, there is a difference in representation, raster vs. vector, which leads to a difference in the design and implementations for the ANN search. Our work considers polygons in vector representation (latitude/longitude list). 

\section{MinHashing for Approximate Similarity}\label{meth}
\begin{figure}[htb]
  \centering
  \includegraphics[width=0.7\linewidth]{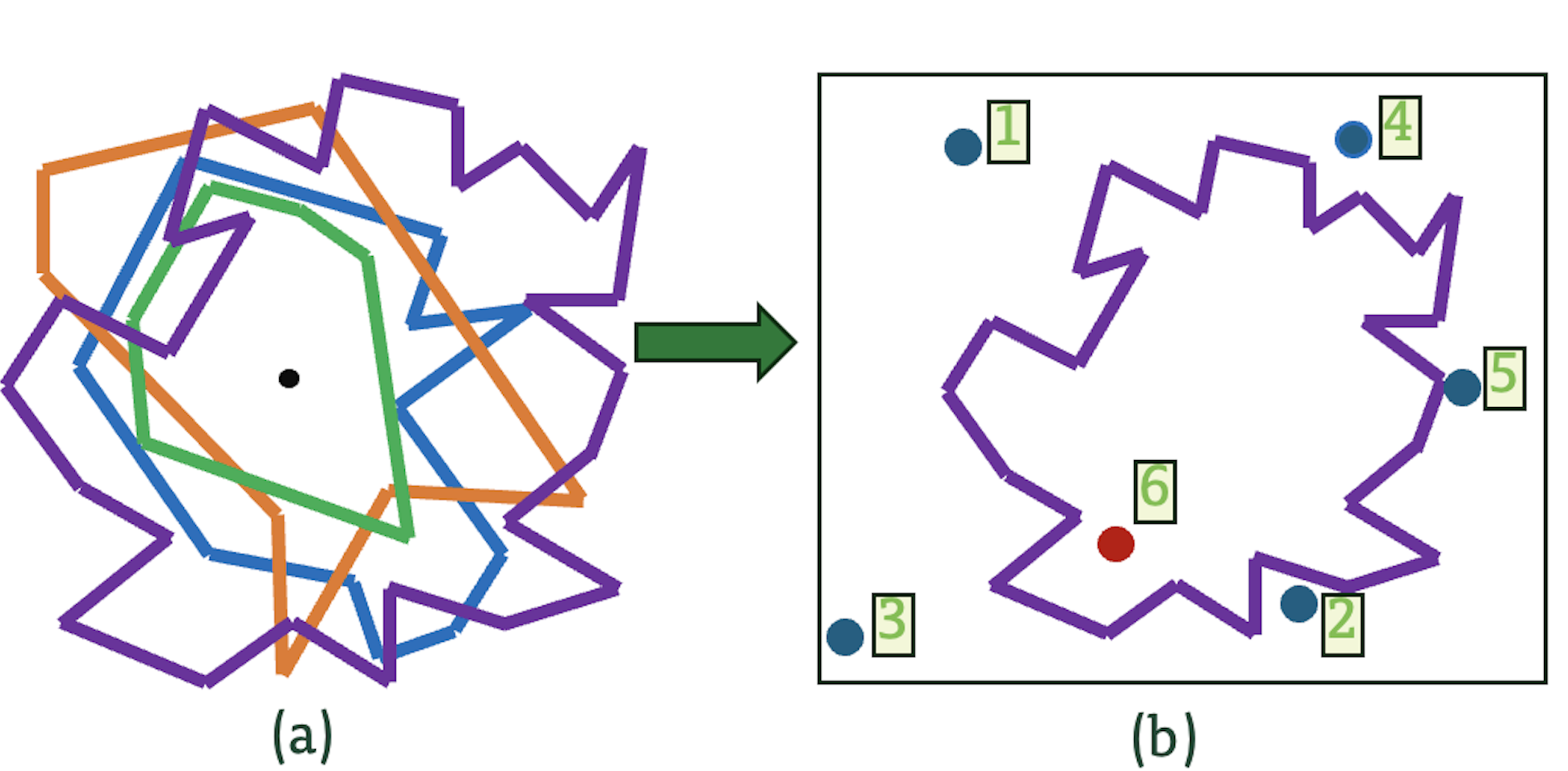}
  \caption{Overview of PolyMinHash system.  
    (a) Center the polygons.  
    (b) Construct the global MBR B and sample 2D points randomly within B for Polygon P. The red dot \#6 is the first point that lands inside the polygon. Therefore, the MinHash(P) is 6. 
   }
  \label{fig:polyhash}
\end{figure}



\subsection{Preprocessing}\label{posk}

The first step in our PolyMinHash system is to create a global MBR that encloses all polygons, which we then use for rejection sampling. Therefore, we apply two important preprocessing steps. Figure ~\ref{fig:polyhash} illustrates the PolyMinHash algorithm.

\textbf{Centering}: We shift each polygon so that its centroid (the geometric center) is positioned at the point $(0, 0)$. This makes it easier to compare shapes that might be located far apart on a real-world map, for example, two parks in completely different regions. 

    
\textbf{Global minimum bounding rectangle}: After centering, we first compute the minimum bounding rectangle (MBR) around each individual polygon. Then, we calculate a global MBR that fully encloses all polygons in the dataset. 

\subsection{PolyMinHash:  MinHashing Scheme for Polygons}\label{Polyhashprop}


Given a polygon $P$ within a global minimum bounding rectangle (MBR) $B$, we first compute the local MBR of $P$,  \([x_{\min}^{P},x_{\max}^{P}]\times[y_{\min}^{P},y_{\max}^{P}]\).
To generate each MinHash value of the signature, we initialize our random number generators using predetermined seeds to ensure consistency across datasets. Fixed seeds make sure that the same sequence of 2D samples will be generated from one run to another. Then we sample a point \((x,y)\) uniformly from the global MBR:
\[
  x \sim \mathrm{Uniform}(x_{\min},x_{\max}),
  \quad
  y \sim \mathrm{Uniform}(y_{\min},y_{\max}).
\]


\begin{definition}[MinHash Function]
	Given a polygon $P$, we define MinHash function $h(P)$ as the number of attempts needed by a random 2D point sampled uniformly over the global MBR to land inside $P$. To generate a MinHash signature of length \(m\), the MinHash function is invoked \(m\) times.
	\label{hs}
\end{definition}




By choosing \(x\) and \(y\) independently and uniformly over their respective intervals, each infinitesimal cell of area \(dx\,dy\) receives the same sampling probability. Once a sample point is generated, several checks are performed sequentially to efficiently filter out points that do not satisfy the following criteria.

\textbf{MBR check:} First, if the sampled point lies outside the local MBR of the polygon, it is immediately rejected without further computation. 

\textbf{Point-in-polygon test:} Otherwise a precise point-in-polygon (PnP) test is performed. If the point lies inside the polygon, the number of attempts taken to land a point inside is recorded as a hash value. 

The pseudocode of the proposed MinHash scheme is shown in Algorithm~\ref{alg:hashing}. This randomized sampling procedure captures the spatial information of a polygonal shape in a concise way.



\begin{algorithm}[H]
	\caption{ PolyMinHash for Polygon Geometry }
	 \begin{algorithmic}[1]
    \Require Polygon vertex array $P$, global bounding box $B$, 
    seed array \texttt{seed[]}, hash length $m$
    \Ensure Hash vector $h[1..m]$
    \State Initialize $h[i]\gets 0$ for all $i$                             
    \State Compute local MBR of $P$                                           \Comment{get $[x^P_{\min},x^P_{\max}]\times[y^P_{\min},y^P_{\max}]$}
    \State Initialize Random Number Generators with each value in \texttt{seed[]}               \Comment{fixed seeds for consistency}

    \For{$i\gets1$ \textbf{to} $m$}
      \State $\mathit{attempts}\gets0$                                      
      \While{true}
        \State $\mathit{attempts}\gets\mathit{attempts}+1$                    
         \State $x\sim\mathrm{Uniform}(x_{\min},x_{\max})$ 

         \State $y\sim\mathrm{Uniform}(y_{\min},y_{\max})$ 
        \If{$x\notin[x^P_{\min},x^P_{\max}]$ \textbf{or} $y\notin[y^P_{\min},y^P_{\max}]$}
          \State \textbf{continue}                                            \Comment{outside polygon’s local MBR}
        \EndIf
        \If{point $(x,y)$ is inside polygon $P$}
          \State $h[i]\gets\mathit{attempts}$                               
          \State \textbf{break}                                             
        \EndIf
      \EndWhile
    \EndFor
    \State \Return $h$
  \end{algorithmic}
    \label{alg:hashing}
\end{algorithm}

\textbf{Correctness of PolyMinHash:} 
Here we want to theoretically explore the question - if two polygons are similar, then what are the chances that they will have the same MinHash code. This is an important property to be satisfied in order to hash similar polygons to the same bucket in a hashmap.

\begin{theorem}\label{thm:collision_probability}
For any two polygons P and Q, the probability of a hash collision equals the Jaccard similarity between the two polygons.

\begin{equation}
\Pr[h(P) = h(Q)] = J(P, Q)
\end{equation}

where \( J(P, Q) \) denotes the Jaccard similarity.

\end{theorem}

\begin{proof}
Intuitively, the first accepted sample point denoted by variable \emph{attempts} (line \#14) represents a randomly sampled point. The probability that $p$ falls within the area of two similar polygons corresponds to their intersection area (area of overlap).




Let \(\{(x_1,y_1),(x_2,y_2),\dots\}\) be the sequence of points independently and uniformly sampled from \(B\). We define the MinHash of polygon $P$ and $Q$ as 
\begin{equation}
  h(P) = \min\{\,j \mid (x_j,y_j)\in P\},
 \;
  h(Q) = \min\{\,j \mid (x_j,y_j)\in Q\}
\end{equation}
and let
\begin{equation}
  j = \min\bigl(h(P),\,h(Q)\bigr).
\end{equation}




We have the following mutually exclusive events in the PolyMinHash sampling process that generate point samples denoted by $p$ = \((x_j,y_j)\):
\begin{enumerate}

    \item \(h(P)=h(Q)=j\) holds if and only if $p$ falls in the intersecting region \(P\cap Q\). In this case, the first sampled point to land in either polygon is shared by both, resulting in a hash collision.

    \item $h(P) > h(Q) = j$ holds if and only if $p$ lands in the area occupied by $Q$ first that does not overlap with $P$, i.e., $p$ is contained in the area given by \(Q\setminus P\). So, the sampling process will continue for $P$ leading to a higher MinHash value for $P$.
    

    \item $h(Q) > h(P) = j$ holds when $p$ lands in the area occupied by $P$ first that does not overlap with $Q$, i.e., $p$ is contained in the area given by \(P\setminus Q\). So, the sampling process will continue for $Q$ leading to a higher MinHash value for $Q$.
\end{enumerate}

In the last two cases, hash collision does not occur which indicates dissimilarity. Since the above-mentioned three cases partition the area denoted by \( P \cup Q \), and each point \( (x_j,y_j) \) is uniformly sampled over \( B \), the probability of a hash collision is:
\[
\Pr[h(P) = h(Q)] = \frac{\text{Area}(P \cap Q)}{\text{Area}(P \cap Q) + \text{Area}(P \setminus Q) + \text{Area}(Q \setminus P)} 
\]
\[
= \frac{\text{Area}(P \cap Q)}{\text{Area}(P \cup Q)} = J(P, Q)
\] 
\end{proof}




Using Theorem 1, we can see that probabilistically similar polygons can have the same MinHash codes, which will lead to hash collisions for similar polygons when we insert the polygons into a hashmap. This helps in clustering similar polygons in a hashmap bucket, which helps in the filtering phase of the ANN search.





\subsubsection{\textbf{Running Time and Variance Analysis:}}

\begin{definition}
Let \( S_{\text{p}} \) denote the effective sparsity of a polygon \( P \) w.r.t. the global MBR:
\begin{equation}\label{spar}
S_{p} = \frac{\text{Area of Polygon}}{\text{Area of Global MBR}}
\end{equation}

\end{definition}

This ratio represents the probability that a point sampled uniformly from the global minimum bounding rectangle (MBR) will fall inside the polygon \( P \). This depends on the area and shape of the polygon in relation to the global MBR.

\begin{theorem} \label{thm:running_time}
The expected number of randomly sampled points required to land inside the polygon \( P \), denoted by $h(P)$, satisfies:
\begin{equation}
\mathbb{E}[h(P)] = \frac{1}{S_{p}}
\label{eq:expected_hits}
\end{equation}
with variance:
\begin{equation}
\text{Var}(h(P)) = \frac{1 - S_{p}}{S_{\text{p}}^2}
\label{eq:variance_hits}
\end{equation}


\end{theorem}

\section{Polygon Similarity Implementation}\label{impl}
The experiments were carried out on a multicore CPU system with 128 physical cores, operating at 2 GHz. We utilized 100 MPI processes for all experiments for distributed memory parallelization. 

The source code and link for the datasets have been made publicly available on GitHub\footnote{ \url{https://github.com/subedialima/PolyMinHash}}.

\section{Experiments}\label{expe}

\subsection{Data Sets}\label{dase}

Our experiments used a variety of real-world polygon datasets, including \emph{Cemetery}, \emph{Sports}, and \emph{Parks} in WKT (Well-Known Text) format, sourced from the UCR STAR website~\cite{ucrstar}. The characteristics of these datasets are presented in Table ~\ref{tab:s1}.

\begin{table}[h]
\setlength{\tabcolsep}{2pt}
\renewcommand{\arraystretch}{0.4}
	\centering 
	\caption{Attributes of the data set.}
	\begin{tabular}{l c c c c r}
		\toprule 
			\textbf{Name} &  \textbf{\#Geometries} &
            \textbf{\#queries} &
            \textbf{Avg. \#Pts} &  \textbf{FileSize}  \\
		\midrule
        \text{Urban Areas}  & \text{11.8K} & \text{3K} &
        \text{95.40} &
        \text{24.50 MB}   \\
			\text{Cemetery} & \text{149K}  & \text{3K} & 
            \text{9.26} &\text{49.10 MB}  \\
             
              \text{Parks}  & \text{300K} & \text{3K} &
              \text{319.00} &
                \text{3750.00 MB}   \\
                \text{Sports}  & \text{1M} & \text{20K}
             &
             \text{12.00} &
             
             \text {511.00 MB}  \\
               
		\bottomrule
	\end{tabular}
	\label{tab:s1}
\end{table}

\subsection{Experimental evaluation of different approaches}
We are using performance metrics such as Recall@k and runtime taken by PolyMinHash for querying input polygon compared to the ground truth as found by the Brute-Force (BF). Recall@k is the proportion of top-$k$ approximate nearest neighbors returned by the LSH search, which are part of the set of actual nearest top-$k$ neighbors returned by the BF search. Possible values range from 0 to 1, with higher values indicating better performance.

\subsubsection{MinHash Signature Length} 
We analyze the impact of MinHash length on recall and runtime by varying values from 1 to 5. When MinHash length is short, a large number of polygons tend to be mapped to the same MinHash code, resulting in higher recall, but also increased runtime due to more candidate comparison in the refinement phase, as illustrated in Fig.~\ref{fig:runtime}, despite the lower MinHashing time. In contrast, the search process is expedited at the cost of accuracy when the MinHash code is longer, since very few neighbors are assigned to the same \(m\)-length MinHash code.
\begin{figure}[ht]
	\centering 
	\includegraphics[width=0.75\linewidth]{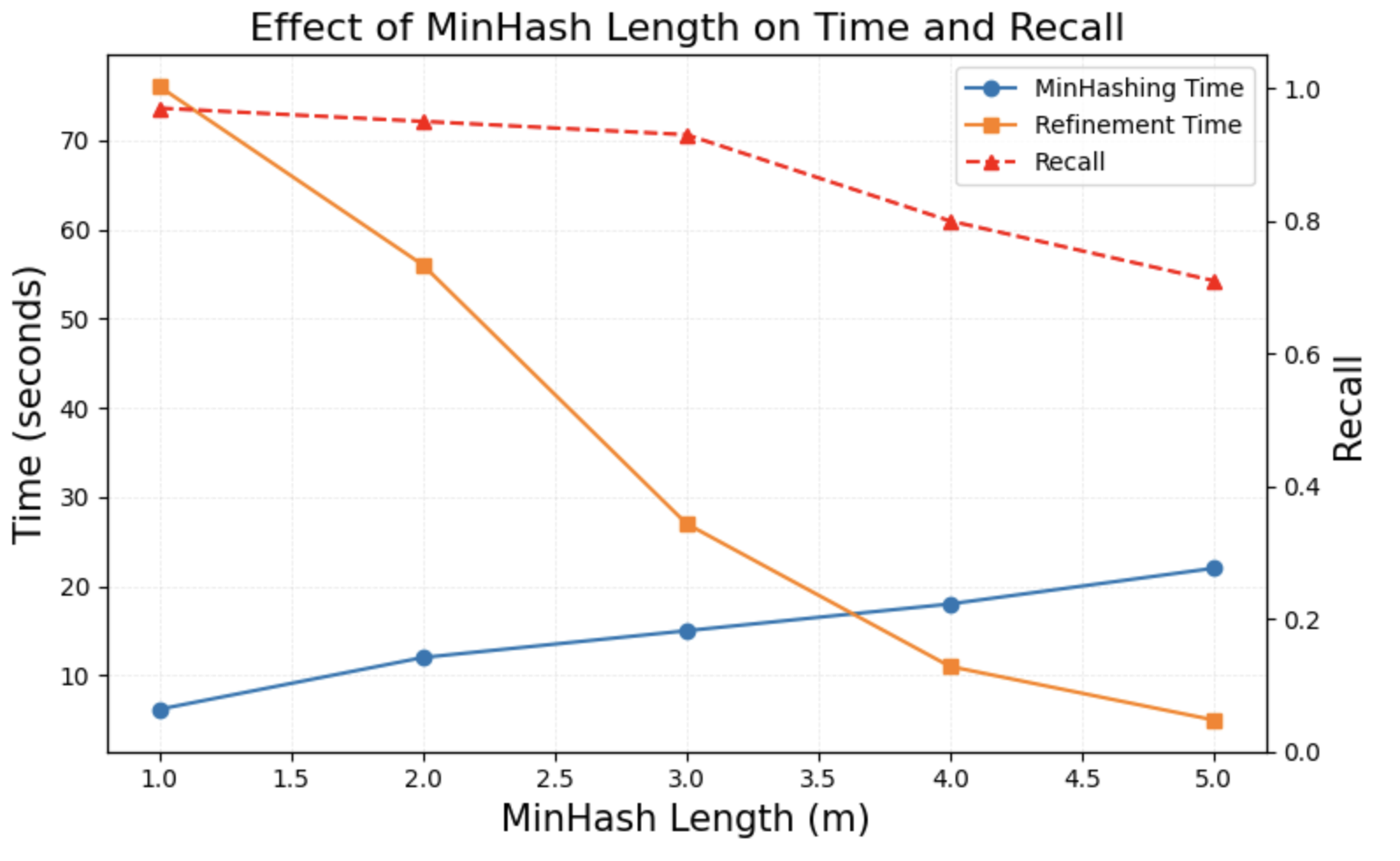}
	
 \caption{Effect of MinHash length on MinHashing time, Refinement time, and Recall for \textit{Cemetery} dataset.}
	\label{fig:runtime}
\end{figure}
Although longer MinHash lengths reduce overall recall, they ensure that only very close neighbors collide in the same MinHash code, so the candidates that are retrieved are highly similar. 



 \begin{table}[ht]
\renewcommand{\arraystretch}{1.0}
    \centering
    \small                          
  \setlength{\tabcolsep}{1pt}

    \textbf{Impact of Different MinHash Lengths on Recall and Runtime in PolyMinHash Across Multiple Datasets. BF: Brute-Force.}

    \caption{
    }

    \resizebox{\columnwidth}{!}{%
    \begin{tabular}{|l|l|l|l|l|l|l|l|l|l|l|}
        \hline
        Dataset &  MinHash  & \multicolumn{3}{c|}{Recall}  & \multicolumn{4}{c|}{Time in Seconds} & Speedup & Pruning (\%) \\
        & Length(m) & $@$10 & $@$50 & $@$100 & \multicolumn{3}{c|}{PolyMinHash} & BF  & BF/Total & \\
        &  &  & & & MinHashing & HashMap Lookup+Refinement & Total &  & & \\
               \hline 
        Cemetery   & 1 & 0.97 & 0.79 & 0.75 & 6.2 & 76.1 & 82.3 &  216.0 & 2.6 & 70 \\
        Cemetery & 3 & 0.93 & 0.88 & 0.81 & 16.1 &  27.5 & 43.6 &216.8  & 4.9 & 89 \\
        Cemetery  & 5 & 0.71 &0.64  & 0.60 & 26.3 & 5.0 & 31.3 &217.2 & 6.9 & 98 \\
        \hline
        Urban Area   & 1 & 0.91 & 0.90 & 0.85 & 3.0 & 13.1 & 16.1 & 32.2 & 2.0 & 70 \\
        Urban Area   & 2 & 0.81 & 0.78 & 0.76 & 4.5 &9.7 & 14.2 & 31.8  & 2.2 & 77 \\
        Urban Area & 3 & 0.76 & 0.70 & 0.66 & 6.2  & 5.7& 11.9 & 31.5  & 2.6 & 88 \\
        \hline 
        Sports & 1 & 0.93 & 0.91 & 0.87 & 29.1 &1696.6 & 1725.7 &4851.8  & 2.8 & 75 \\
        Sports & 3 &0.88  & 0.85 &0.80 &  84.0 &1185.9 & 1269.9& 4851.0  &3.8 & 86\\
        \hline 
        Parks & 1  & 0.93  &  0.92& 0.91  & 81.8 & 1900.0 &1981.8 & 6291.7 & 3.1 & 45 \\
        Parks & 2  &0.84& 0.81 & 0.79 & 240.1 &  1260.9 &  1501.0 &6292.6 & 4.1 & 69\\
        Parks & 3 & 0.69 &0.66  & 0.65 & 409.3 & 734.1 &1143.4 & 6292.7  & 5.5 & 86 \\
        \hline 
    \end{tabular}
  }
    \label{tab:GSp}
\end{table}

\subsection{PolyMinHash Filtering Analysis}

Figure ~\ref{fig:pruning}(a) shows how PolyMinHash’s filtering stage reduces candidate sets. It shows PolyMinHash prunes over 86\% of the dataset with good recall (\(\ge0.69\)) and still prunes around 70\% even at very good recall (\(\ge0.84\)). 
Similarly, Fig. ~\ref{fig:pruning}(b) shows the impact of hash length on pruning. Increasing the number of hash functions \(m\) improves the effectiveness of the filter: larger values of \(m\) produce higher pruning ratios by generating more discriminative signatures. 

\begin{figure}[ht]
	\centering 
	\includegraphics[width=1\linewidth]{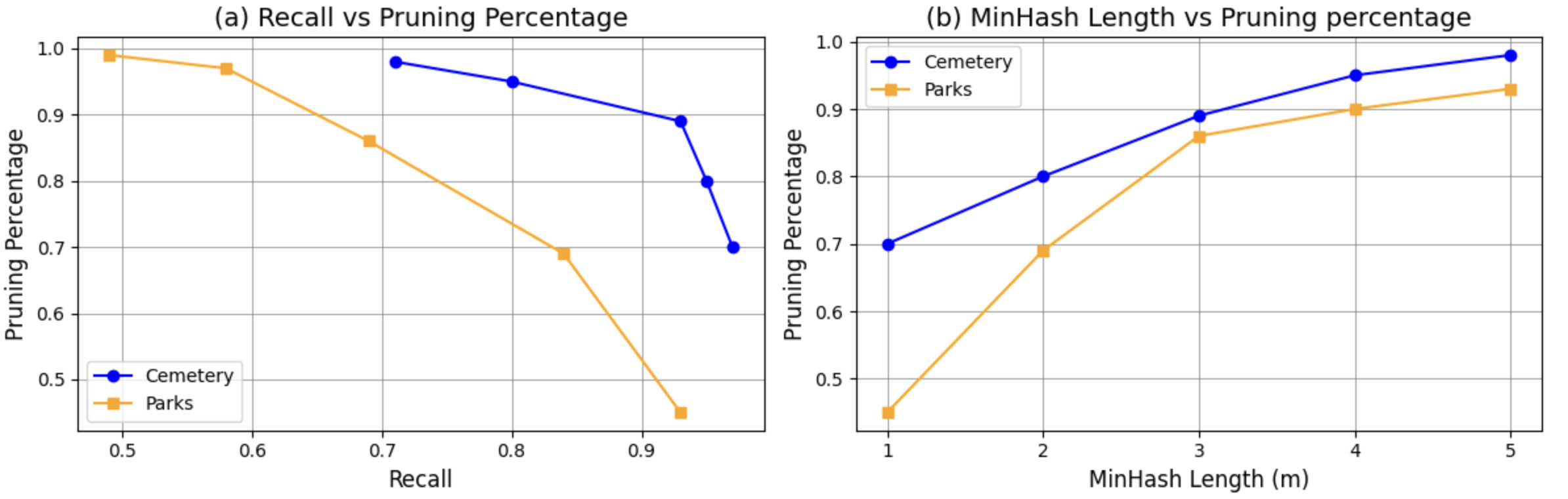}
	
 \caption{(a) Relation between Recall and Pruning. (b) Effect of MinHash Length on Pruning.
}
	\label{fig:pruning}
\end{figure}

The advantage of PolySS system using 2 hashmaps over an exact similarity-based system is demonstrated in Table~\ref{tab:GSp} with different values for Recall@k.


\section{Conclusions}\label{conc}

In this work, we have extended the classical MinHash to operate directly on 2D geometric shapes. Our PolyMinHash framework preserves the Jaccard similarity over polygons, saves memory space in representing polygons concisely, and helps in pruning in ANN search. 





\bibliographystyle{ACM-Reference-Format}
\bibliography{polySS}

@article{pan2024survey,
  title={Survey of vector database management systems},
  author={Pan, James Jie and Wang, Jianguo and Li, Guoliang},
  journal={The VLDB Journal},
  volume={33},
  number={5},
  pages={1591--1615},
  year={2024},
  publisher={Springer}
}

@inproceedings{grauman2004fast,
  title={Fast contour matching using approximate earth mover's distance},
  author={Grauman, Kristen and Darrell, Trevor},
  booktitle={CVPR},
  volume={1},
  pages={},
  year={2004},
  organization={IEEE}
}

@inproceedings{astefanoaeiMultiresolutionSketchesLocality2018,
  title = {Multi-Resolution Sketches and Locality Sensitive Hashing for Fast Trajectory Processing},
  booktitle = {SIGSPATIAL},
  author = {Astefanoaei, Maria and Cesaretti, Paul and Katsikouli, Panagiota and Goswami, Mayank and Sarkar, Rik},
  year = {2018},
  pages = {},
  publisher = {{ACM}},
  location = {},
  doi = {},
  url = {},
  urldate = {2022-06-03},
  eventtitle = {{{SIGSPATIAL}} '18: 26th {{ACM SIGSPATIAL International Conference}} on {{Advances}} in {{Geographic Information Systems}}},
  isbn = {978-1-4503-5889-7},
  langid = {english}
}

@book{leskovecMiningMassiveDatasetsa,
  title = {Mining of {{Massive Datasets}}},
  author = {Leskovec, Jure and Rajaraman, Anand and Ullman, Jeffrey D},
  month = jul,
  year = {2019},
  edition = {3},
  url = {http://infolab.stanford.edu/~ullman/mmds/book0n.pdf},
  langid = {english},
  pagetotal = {603}
}

@inproceedings{shrivastavaSimpleEfficientWeighted2016,
  title = {Simple and {{Efficient Weighted Minwise Hashing}}},
  booktitle = {NIPS},
  author = {Shrivastava, Anshumali},
  year = {2016},
  volume = {29},
  publisher = {},
  url = {},
  urldate = {2022-06-24}
}

@article{DBLP:journals/corr/DriemelS17,
  author       = {Anne Driemel and
                  Francesco Silvestri},
  title        = {Locality-sensitive hashing of curves},
  journal      = {CoRR},
  volume       = {abs/1703.04040},
  year         = {2017},
  url          = {http://arxiv.org/abs/1703.04040},
  eprinttype    = {arXiv},
  eprint       = {1703.04040},
  bibsource    = {dblp computer science bibliography, https://dblp.org}
}

@article{ucrstar,
  author      = {Ghosh, Saheli and Vu, Tin and Eskandari, Mehrad Amin and Eldawy, Ahmed},
  title       = {{UCR-STAR: The UCR Spatio-Temporal Active Repository}},
  journal     = {SIGSPATIAL Special},
  volume      = {},
  number      = {},
  pages       = {},
  year        = {2019},
  doi         = {},
}

@inproceedings{wang2012accelerating,
  title={Accelerating pathology image data cross-comparison on cpu-gpu hybrid systems},
  author={Wang, Kaibo and Huai, Yin and Lee, Rubao and Wang, Fusheng and Zhang, Xiaodong and Saltz, Joel H},
  booktitle={VLDB},
  volume={5},
  number={11},
  pages={1543},
  year={2012},
  organization={}
}

@inproceedings{pillai2013filter,
  title={A filter-and-refine approach to mine spatiotemporal co-occurrences},
  author={Pillai, Karthik Ganesan and Angryk, Rafal A and Aydin, Berkay},
  booktitle={SIGSPATIAL},
  pages={104--113},
  year={2013}
}
\end{document}